\documentclass[12pt]{amsart}
\usepackage[USenglish]{babel}
\usepackage{times,enumitem}
\usepackage{amssymb,amsmath}
\usepackage{graphicx}
\usepackage{cite}
\usepackage{amsaddr}
\usepackage[usenames]{color}
\usepackage{hyperref}

\newtheorem{definition}{Definition}

\newtheorem{theorem}{Theorem}
\newtheorem{lemma}[theorem]{Lemma}

\newtheorem*{theorem*}{Theorem}

\title[Cryptanalysis of a Cayley Hash Function]{Cryptanalysis of a Cayley Hash Function Based on\\ 
Affine Maps in one Variable over a Finite Field}

\author[]{Bianca Sosnovski} 
\thanks{The author received support for this project provided by a PSC-CUNY grant, jointly funded by the Professional Staff Congress and the City University of New York.}
\address{Department of Mathematics and Computer Science\\ Queensborough Community College, CUNY}
\email{bsosnovski@qcc.cuny.edu}

\begin{document}

\maketitle

\begin{abstract}

The hash function proposed by Shpilrain and Sosnovski (2016), based on affine maps in one variable over a finite field, was proven insecure. This paper shows that the variation proposed by Ghaffari and Mostaghim (2018) that uses Shpilrain and Sosnovski's hash is also insecure. We demonstrate its security vulnerability by constructing collisions.

\

\noindent \textbf{Keywords:} Cryptography, hash functions, Cayley hash functions, cryptanalysis, collision attack.

%

\end{abstract}

\section{Introduction}

Hash functions are an essential tool for cryptography. Today's security of much of our communication relies on cryptographic protocols that ensure confidentiality, integrity and authentication, and many such protocols use hash functions as building blocks. Hash functions are fundamental in constructing cryptographic protocols, such as database indexing, data compression, password storage, digital signatures, encryption schemes, and key derivation systems.

However, not every hash function is good enough for cryptography. Cryptographic hash functions are hash functions that satisfy desired security properties such as preimage and collision resistance and may be used in cryptographic applications. 

Furthermore, many cryptosystems in use today are based on finite abelian groups. Some cryptographic systems will be vulnerable to attacks once large quantum computers are made possible. Though the current state of quantum computing is still in its infancy, it is a step forward in the direction where classical cryptography may be compromised. In \cite{bernstein09}, hash-based public-key signatures are one of the classes of cryptographic systems that may resist quantum attacks and require a standard cryptographic hash function.

Provably secure hash functions are hashes whose security is implied by the assumption of the hardness of a mathematical problem. Examples of provable-secure hash functions are the Cayley hash functions. Cayley hash functions are families of hash functions constructed from Cayley graphs of the groups \cite{plq}. The security of Cayley hash functions would follow from the alleged hardness of a mathematical problem related to the Cayley graph regarding a generating set of the underlying group \cite{cgl09, plq}.  Since Cayley hashes involve non-abelian groups, it is a priori resistant to quantum attacks, and they may be good candidates for post-quantum cryptography \cite{jy18}.

In 1991, Z\'emor introduced the first  Cayley hash function \cite{zemor91} that has as generators the matrices 
$$
\left(
\begin{array}{cc}
 1 &   1   \\
  0&1 
\end{array}
\right)
\mbox{ and }
\left(
\begin{array}{cc}
 1 &   0   \\
  1&1 
\end{array}
\right)
$$
\noindent and its hash values are elements in $SL_{2}(\mathbb{F}_{p})$ for $p$ prime. 

It was broken by Tillich and Z\'emor in 1994 \cite{tz94group}, who then proposed the hash function whose generators  
$$\left(
 \begin{array}{cc} \alpha & 1 \\ 1 & 0 \end{array} \right)  \mbox{ and }  \left(
 \begin{array}{cc} \alpha & \alpha+1 \\ 1 & 1 \end{array} \right)$$ 
 with $\alpha$ as the root of an irreducible polynomial $p(x)$ of degree $n$ in the ring of polynomials $\mathbf{F}_2[x]$, where $\mathbf{F}_2$ is the field with
two elements \cite{tz94}. The above matrices are generators of the Cayley graph for the group $SL_2(\mathbb{F}_{2^n})$ with $\mathbb{F}_{2^n}\approx \mathbf{F}_2[x]/(p(x))$ where $(p(x))$ is the ideal generated by an irreducible polynomial $p(x)$.

The Tillich-Z\'emor hash function was broken in 2009 when Grassl et al. \cite{gims11} established a connection between the Tillich-Z\'emor function and maximal length chains in the Euclidean algorithm for polynomials over the field with two elements.  Other instances of Cayley hashes based on expander graphs have been proposed after Tillich-Z\'emor functions. Detailed discussions of Cayley hash functions can be found in \cite{petit09, plq08, plq, cgl09}. These Cayley hashes also have been proven insecure.

Though many instances of Cayley hash functions have been proved insecure, the algorithms used to break Cayley hash functions target specific vulnerabilities of each underlying group used and do not invalidate the generic scheme of these functions. The factorization, representation and balance problems in non-abelian groups still are potentially hard problems for general parameters of Cayley hash functions. There are still Cayley hash functions that remain unbroken (e.g., \cite{bsv17,cbf22, yuan16}).

It may seem a concerning scenario where many hashes have been proven insecure. But this is also essential and encouraging in cryptography since it demonstrates that the community invests a lot of time and energy in cryptanalysis to ensure algorithms are evaluated and that new ones are developed to sustain quantum attacks. The more researchers and scientists have looked at these algorithms and they remain unbroken, the higher our level of confidence in them.

This paper proves that the hash function proposed by Gaffari and  \linebreak Mustaghim \cite{gm18} is not collision-resistant, which uses the hash proposed by Shpilrain and Sosnovski \cite{ss16} that has been proven insecure by Monico \cite{monico18}. To show that  Gaffari and Mustaghim's is also insecure, we apply Monico's algorithm to find second-preimages for the Shpilrain-Sosnovski hash function to produce collisions for the Gaffari and Mustaghim's hash function. 

The remainder of the paper is organized as follows. In Section \ref{sec2}, we recall some basic definitions and properties of a cryptographic hash function. Section \ref{sec3} briefly describes the Shpilrain-Sosnovski hash and Gaffari-Mustaghim hash. Section \ref{sec4} presents a summary of the cryptanalysis of the Shpilrain and Sosnovski's hash function. In Section \ref{sec5}, we present our main results about the security of the Gaffari and Mustaghim's hash:

\begin{theorem*} Ghaffari-Mostaghim hash is not collision-resistant.
\end{theorem*}

\section{Preliminaries}\label{sec2}

Hash functions are used as compact representations, or digital fingerprints, of data to provide message integrity.  

\begin{definition}
A hash function $h: \{0,1\}^{*}\longrightarrow \{0,1\}^{n}$ is an easy-to-compute\footnote{Easy to compute or computationally feasible means polynomial time and space or, in practice, with a certain number of machine operations to time units \cite{menezes}.} function that converts a variable-length input into a fixed-length output. A cryptographic hash function $h$ must satisfy at least one of the following properties.

\begin{itemize}
\item \emph{Preimage resistance}: Given a hash value $y$ for which a corresponding input is not known,  it is computationally infeasible (or hard) to find any input $x$ such that $y=h(x)$. 
\item  \emph{Second-preimage resistance}: Given an input $x_1$ it is computationally infeasible to find another input $x_2$ where $x_1 \neq x_2$ such that $h(x_1)=h(x_2)$. 
\item  \emph{Collision resistance}: It is computationally infeasible to find any two inputs $x_1$ and $x_2$ where $x_1 \neq x_2$ such that $h(x_1)=h(x_2)$.
\end{itemize}

\end{definition}

A collision-resistance hash function is also second-preimage resistant. Preimage resistance does not guarantee second-preimage resistance, and Second-preimage resistance does not ensure preimage resistance \cite{menezes}.

It is well known that expander graphs are used to produce pseudorandom behavior. This pseudorandom behavior is due to the rapid mixing of Markov chains on expander graphs. The initial idea was to use groups whose Cayley graphs concerning a set of generators are expander graphs to design collision-resistant hash functions.

\begin{definition}
Let $G$ be a finite group with a set of generators $\mathcal{S}$ that has the same size as the text alphabet\footnote{In general, we can consider plaintexts as strings of symbols from a text alphabet $\{1,2, \ldots, k\}$ for $k\geq 2$. Conventionally, we use the text alphabet as $\{0,1\}$ for binary strings.} $\mathcal{A}$. Choose a function:
$\pi : \mathcal{A}\to \mathcal{S}$ such that $\pi$ defines a one-to-one correspondence between $\mathcal{A}$ and $\mathcal{S}$. A  \emph{Cayley hash} $h$ is a function whose hash value of the text $x_{1}x_{2}\dots x_{k}$ is the group element $h(x_{1}x_{2}\dots x_{k})=\pi(x_{1})\pi(x_{2})\dots \pi(x_{k})$. 

\end{definition}

For example, a Cayley hash has  $A$
 and $B$ as the generators of the underlying group $G$ with the bit assignments $0 \mapsto A$ and $1 \mapsto B$. The bit string 101011 is hashed to the group product $BABAB^2$.

In constructing hash functions from expander Cayley graphs, the input to the hash function gives directions for walking around the graph (without backtracking), and the hash output is the end vertex of the walk.

For the Cayley hashes described in this paper, the alphabet used corresponds to $\{0,1\}$. One of the advantages of this design is that the computation of the hash value can be easily parallelized due to the concatenation property $\pi(xy)= \pi(x)\pi(y)$ for any texts $x$ and $y$ in $\{0,1\}^{*}$. Unlike the SHA family of hash functions that hash blocks of input, this type of function hashes each bit individually. 

The security properties of Cayley hash functions are strongly related to the hardness of mathematical problems. 

Let $G$ be a group and $\mathcal{S}=\{s_{1}, \ldots s_{k}\}\subset G$ be a generating set of $G$. 
Let $L$ be polylogarithmic (small) in the size of $G$. 

\begin{itemize} 
\item \emph{Balance problem:} Find an efficient algorithm that returns two words $m_{1} \ldots m_{l}$  and $m'_{1} \ldots m'_{l'}$ with $l,l'<L$, $m_{i},m'_{i}\in \{1,\ldots, k\}$ that yield equal products in $G$, that is, $\prod \limits_{i=1}^l s_{m_{i}}=\prod \limits_{i=1}^{l'} s_{m'_{i}}$
\item \emph{Representation problem:} Find an efficient algorithm that returns a word $m_{1} \ldots m_{l}$  with $l<L$, $m_{i}\in \{1,\ldots, k\}$ such that $\prod \limits_{i=1}^l s_{m_{i}}=1$.
\item \emph{Factorization problem:} Find an efficient algorithm that given any element $g\in G$ returns a word $m_{1} \ldots m_{l}$  with $l<L$, $m_{i}\in \{1,\ldots, k\}$ such that $\prod \limits_{i=1}^l s_{m_{i}}=g$.
\end{itemize}

A Cayley hash function is collision-resistant if the balance problem is hard in the underlying group. Suppose the representation problem is hard in the group. In that case, the associated Cayley hash is second preimage resistant, and it is preimage resistant if and only if the corresponding factorization problem is hard in the group \cite{petit09,pq11v2}.

Other requirements considered by Tillich and Z\'emor \cite{zemor91, tz94} in the construction of Cayley hash functions are that the Cayley graph of G with generator set S has a large girth and small diameter.

\section{Cayley hash functions}\label{sec3}

%
%

\subsection{The Shpilrain-Sosnovski hash function}

In \cite{ss16}, the authors presented a Cayley hash function that uses linear functions in one variable over $\mathbb{F}_{p}$ with composition operation. 

The semigroup generated  by $f(x)=ax+b$  and $g(x)=cx+d$  under composition is isomorphic to the semigroup generated by 
$$
A = \left(
 \begin{array}{cc} a & b \\ 0 & 1 \end{array} \right) \mbox{ and } B = \left(
 \begin{array}{cc} c  & d \\ 0 & 1 \end{array} \right)
$$
under matrix multiplication. Using results about the freeness of upper triangular matrices by Cassaigne at al. \cite{chk99}, they showed that the semigroup of linear functions over $\mathbb{Z}$ is free if the generators of the semigroup do not commute and $a,c\geq 2$.

The functions $f_{0}(x)=2x+1  \mod p$ and $f_{1}(x)=3x+1 \mod p$ with  $p>3$ are considered the generators of the proposed hash function. The hash value is obtained by first computing product $h( b_1 b_2 \cdots b_k)=f_{b_1}f_{b_2} \cdots f_{b_k} \pmod p$ where $b_i\in \{0,1\}$ for $1\le i\le k$. The  corresponding product linear function is of the form $\ell(x)=r x+s$ where $s,r\in\mathbb{Z}_{p}$, and the hash value is defined as  $H( b_1 b_2 \cdots b_k)= (r+s,s)$. 

The corresponding hash functions are very efficient. A bit string of length $n$ can be hashed by performing at most $2n$ multiplications and about $2n$ additions in $\mathbb{F}_{p}$.


An advantage of this hash function is that the output bit strings have length $2\log p$,  while the Tillich-Z\'emor hash function outputs bit strings of length $4\log p$. Concerning the security of the hash function, the authors recommend that $p\approx 2^{512}$ or larger to prevent generic attacks. With this recommended parameter  $p$, there will be no collisions unless the length of at least one of the colliding strings is at least 323. For a short input text (323 bits or less), the authors recommend padding to extend its length to 512 bits. Subgroup attacks and attacks using elements of small orders can be prevented by choosing  $p$ such that $p=2q+1$ where $q$ is a ``large'' prime. 

\subsection{The Ghaffari-Mostaghim hash function}

As discussed in \cite{ss16}, preimages can be easily computed for short messages in the Shpilrain-Sosnovski hash, and the option suggested to avoid it is using padding.

Ghaffari and Mostaghim use a similar idea introduced in \cite{petit09} to modify the linear hash function above. The functions $f_{0}(x)=2x+1  \mod p$ and $f_{1}(x)=3x+1 \mod p$, where $p>3$ is a prime, are also considered as generators in this Cayley hash.  Let $H(m_1 m_2 \cdots m_l)= f_{m_1}f_{m_2} \cdots f_{m_l} \pmod p$ for $m=m_1 m_2 \cdots m_l\in\{0,1\}^{*}$.
Define the new function $H_2 (m)=H(m\parallel (H(m) \oplus c_{rnd} ))$,  where $c_{rnd}$ is a constant bit string whose bits look random. $H_2$ is meant to be a more secure version of $H$, especially for short messages, and also avoids the issue of malleability.

To make the factorization problem harder, Ghaffari and Mostaghim \cite{gm18} suggested the following variation. Let $G$ the group generated by $f_{0}$ and $f_{1}$ over $\mathbb{Z}_{p}$,  $t>1$ an integer and  $g\in G \setminus \{e, f_{0}, f_{1}\}$,  where $e$
is the identity element of $G$. 

Define $\widehat{H} : \{0, 1\}^{*}\to G$ by $\widehat{H} (m) =\prod_{i=1}^{l} C_{i}$
 where  
$$C_{i}=\begin{cases} f_{m_{i}} &\mbox{if } t \nmid i \\ 
f_{m_{i}} g & \mbox{if } t\mid i \end{cases}.$$
Now define $\widehat{H}_{2}(m)=\widehat{H}(m \parallel (\widehat{H}(m)\oplus c_{rnd}))$.

For an input bit string of length $l$, the computation of $\widehat{H}$ requires $\lfloor l/t \rfloor$ multiplications more than the original Cayley hash function proposed by Shpilrain and Sosnovski, thus not affecting too much the performance of the hash.

Ghaffari and Mostaghim showed that $\widehat{H}$ is at least as secure as the Shpilrain-Sosnovski hash function $H$, and consequently, so is $\widehat{H}_{2}$. 

\section{Monico's Algorithm}\label{sec4}

Monico \cite{monico18} developed an attack that shows that the hash function is not second-preimage resistant for inputs larger than about 1.9 MB for parameter $p\approx 2^{256}$. In Monico's method, the original bit string is not even required, and having only a bound on its length suffices (preimage weakness).

In Monico's attack, a hash value $(x,y)$ in $\mathbb{F}_{p}^{2}$ of a bit string of known length $L$ is given and inverted to $(r,s)=(x-y,y)$. Since $r=2^{a}3^{b}$ where $a$ is the number of zeros in the original bit string, and $b$ is the number of ones (or vice-versa), then $L=a+b$. The values of $a$ and $b$ can be recovered with $O(L \log L)$ operations over $\mathbb{F}_{p}$ by precomputing $L$ powers of $2$, sorting them out and then computing and testing $r, 3^{-1}r, 3^{-2}r, \ldots$ until one of the values in the sequence matches one of the precomputed powers of 2. 

Let $n=min\{a,b\}$, 
$$Y = \left(
 \begin{array}{cc} r & s \\ 0 & 1 \end{array} \right) \mbox{ and } U = \left(
 \begin{array}{cc} r & u \\ 0 & 1 \end{array} \right), $$ 
 where $U$ is a suitable matrix whose factorization in  generators 
 $$A =\left(
 \begin{array}{cc} 2 & 1 \\ 0 & 1 \end{array} \right)  \mbox{ and } B = \left(
 \begin{array}{cc} 3 & 1 \\ 0 & 1 \end{array} \right) $$ 
 is known and determined by the values of $a$ and $ b$ found in the first step. 
  
 The attack aims to transform $U$ into $Y$ by replacing several of the leading $AB$ factors of $U$ with $BA$. To do so, one must find  $\mathbf{x}\in \{0,1\}^{n}$ such that $\displaystyle \sum_{j=0}^{n-1} x_{j}6^{j} \equiv t \pmod p$ where $t=s-u \pmod p$ (for more details, see \cite{monico18})

To provide a probabilistic algorithm to find such $\mathbf{x}$, Monico reduced the problem to a dense instance of the Random Modular Subset Sum Problem (RMSSP), which was considered by Lyubashevsky (2005). Heuristically, his algorithm is expected to succeed as long as the original bit string had at least $n$ zeros and $n$ ones for some $ n\ge 2^{\sqrt{2\log_{2} p}}$. According to Monico, the algorithm's expected running time is $O(n^{2} \log n)$ with an implied constant small enough to keep the attack practical for $p\approx 2^{256}$.

%

%

%

\section{Cryptanalysis of the Ghaffari-Mostaghim hash}\label{sec5}

This section uses Monico's algorithm to produce collisions for the Ghaffari-Mostaghim hash function

\begin{lemma}\label{lemma1} 
A collision for $\widehat{H}$ is also a collision for $\widehat{H}_2$.
\end{lemma}

\begin{proof}
Suppose that $m$ and $m'$ are two bit strings such that $\widehat{H}(m)=\widehat{H}(m')$.

\[
\begin{array}{ccc}
\widehat{H}_2(m)  &  =  & \widehat{H}\left(m\parallel (\widehat{H}(m)\oplus c_{rnd})\right)   \\
  &  = & \widehat{H}(m)\widehat{H}(\widehat{H}(m)\oplus c_{rnd})    \\
  & =  &   \widehat{H}(m')\widehat{H}(\widehat{H}(m')\oplus c_{rnd})  \\
  &  = & \widehat{H}\left(m'\parallel (\widehat{H}(m')\oplus c_{rnd})\right)   \\
  & =  &  \widehat{H}_2(m')
\end{array}
\]

\end{proof}

\begin{lemma}\label{lemma2} A collision for $H$ is also a collision for $H_2$.
\end{lemma}

\begin{proof} Similar to the proof in Lemma \ref{lemma1}.

%

\end{proof}

\begin{theorem} Ghaffari-Mostaghim hash is not collision-resistant.
\end{theorem}

\begin{proof}
Monico's algorithm can find collisions for $\widehat{H}_2$.

Let $t>1$ and $g\in G$ as described in the construction of $\widehat{H}_2$. We can use the algorithm to find a preimage for $g^{-1}$ under the hash $H$, say $b'$ such that $H(b')=g^{-1}$.

Suppose that for a given bit string $m= m_1 m_2 \cdots m_{l_1} $, Monico's algorithm returns a bit string $m' = m_1' m_2' \cdots m_{l_2}'$ such that $H(m) = H(m')$ with $m \neq m'$.
We insert $b'$  into the bit strings $m$ and $m'$ in the bit positions multiple of $t+1$, obtaining the following
$$m_{*}= m_1 m_2  \cdots   m_{t}  b' \cdots m_{l_1}$$
$$m_{*}'= m_1' m_2'  \cdots   m_{t}'  b' \cdots m_{l_2}'$$

We have that $m_{*} \neq m_{*}'$ and they are collisions for $\widehat{H}$ since
\[ 
\begin{array}{lll}
\widehat{H}(m_{*}) & = &\widehat{H}(m_1 m_2  \cdots   m_{t}  b' \cdots m_{l_1})\\
                               & = & H(m_1) H(m_2 ) \cdots   H(m_{t}) H( b' )\cdots H(m_{l_1})\\
                               & = & f_{m_1} f_{m_2 } \cdots   f_{m_{t}} g H( b' )\cdots f_{m_{l_2}}\\
                                & = & f_{m_1} f_{m_2 } \cdots   f_{m_{t}} g g^{-1}\cdots f_{m_{l_1}}\\
                                 & = & H(m)
\end{array}
\]

\[ 
\begin{array}{lll}
\widehat{H}(m_{*}') & = &\widehat{H}(m_1' m_2'  \cdots   m_{t}'  b' \cdots m_{l_2}')\\
                               & = & H(m_1') H(m_2 ') \cdots   H(m_{t}') H( b' )\cdots H(m_{l_2}')\\
                               & = & f_{m_1'} f_{m_2' } \cdots   f_{m_{t}'} g H( b' )\cdots f_{m_{l_2}'}\\
                                & = & _{m_1'} f_{m_2 '} \cdots   f_{m_{t}' }g g^{-1}\cdots f_{m_{l_2}'}\\
                                 & = & H(m')
\end{array}
\]
%
%
%

This shows that Monico's algorithm produces collisions for $\widehat{H}$. Therefore, $\widehat{H}_2$ is not collision-resistant by Lemma \ref{lemma1}. 

\end{proof}
\


\section{Conclusion}\label{sec6}

This paper proves that the variant proposed by Ghaffari and Mostaghim is insecure. Our approach is to consider the mathematical structure of the design of the hash and apply the algorithm that finds second preimages and collisions for the Shpilrain and Sosnovski's hash function.

The algorithms used to break Cayley hash functions target specific vulnerabilities of each underlying group used and do not invalidate the generic scheme of these functions.  Petit and Quisquater \cite{pq11,pq11v2} suggested that security might be recovered for the Cayley hash functions design by introducing new generators. 

Although many Cayley hash functions have been proven insecure, and their use significantly compromises the security of computer systems, we learn from prior vulnerabilities to develop more robust hash functions. It is essential to research improvements and develop new designs for Cayley hash functions that can sustain quantum attacks.


%
%
%
%

\end{document}